\documentclass[journal]{IEEEtran}
\IEEEoverridecommandlockouts
\usepackage[utf8]{inputenc}
\usepackage{cite}
\usepackage{amsmath,amssymb,amsfonts}
\usepackage{graphicx}
\usepackage{textcomp}
\usepackage{xcolor}
\usepackage{epstopdf}
\usepackage{amssymb}
\usepackage{amsmath}
\usepackage{algorithm}
\usepackage{algorithmicx}
\usepackage{algpseudocode}
\usepackage{amsmath}
\usepackage{multirow}
\usepackage{amsthm}
\usepackage{float}
\theoremstyle{plain}
\newtheorem{thm}{Theorem}

\def\BibTeX{{\rm B\kern-.05em{\sc i\kern-.025em b}\kern-.08em
    T\kern-.1667em\lower.7ex\hbox{E}\kern-.125emX}}
\makeatletter
\long\def\@makecaption#1#2{\ifx\@captype\@IEEEtablestring%
\footnotesize\begin{center}{\normalfont\footnotesize #1}\\
{\normalfont\footnotesize\scshape #2}\end{center}%
\@IEEEtablecaptionsepspace
\else
\@IEEEfigurecaptionsepspace
\setbox\@tempboxa\hbox{\normalfont\footnotesize {#1.}~~ #2}%
\ifdim \wd\@tempboxa >\hsize%
\setbox\@tempboxa\hbox{\normalfont\footnotesize {#1.}~~ }%
\parbox[t]{\hsize}{\normalfont\footnotesize \noindent\unhbox\@tempboxa#2}%
\else
\hbox to\hsize{\normalfont\footnotesize\hfil\box\@tempboxa\hfil}\fi\fi}
\makeatother

\begin{document}

\title{Edge-Cloud Collaboration Enabled Video Service Enhancement: A Hybrid Human-Artificial Intelligence Scheme}
\author{Dapeng Wu,~\IEEEmembership{Senior Member, IEEE,}
        Ruili Bao,~
        Zhidu Li,~\IEEEmembership{Member, IEEE,}\\
        Honggang Wang,~\IEEEmembership{Senior Member, IEEE,}
        Hong Zhang,
        and Ruyan Wang
        \thanks{Dapeng Wu, Ruili Bao, Zhidu Li, Hong Zhang and Ruyan Wang are with the School of Communication and Information Engineering, Chongqing University of Posts and Telecommunications, Chongqing 400065, China, and with Key Laboratory of Optical Communication and Networks in Chongqing, and with Key Laboratory of Ubiquitous Sensing and Networking in Chongqing (e-mail: $\{\text{wudp};\text{s190131250};\text{lizd};\text{hongzhang};\text{wangry}\}$@cqupt.edu.cn).}
        \thanks{Honggang Wang is with Electrical and Computer Engineering Department, University of Massachusetts Dartmouth, USA (e-mails: hwang1@umassd.edu).}
}

\maketitle

\begin{abstract}
In this paper, a video service enhancement strategy is investigated under an edge-cloud collaboration framework, where video caching and delivery decisions are made in the cloud and edge respectively.
We aim to guarantee the user fairness in terms of video coding rate under statistical delay constraint and edge caching capacity constraint.
A hybrid human-artificial intelligence approach is developed to improve the user hit rate for video caching.
Specifically, individual user interest is first characterized by merging factorization machine (FM) model and multi-layer perceptron (MLP) model, where both low-order and high-order features can be well learned simultaneously.
Thereafter, a social aware similarity model is constructed to transferred individual user interest to group interest, based on which, videos can be selected to cache.
Furthermore, a double bisection exploration scheme is proposed to optimize wireless resource allocation and video coding rate.
The effectiveness of the proposed video caching scheme and video delivery scheme is finally validated by extensive experiments with a real-world data set.
\end{abstract}


\section{Introduction}
With the technology evolutions in communication, computing and caching, there are huge amount of services emerging and posing a profound influence on our daily life \cite{JNAVARRO}.
In particular, video services, such as live streaming, video on demand and etc., are more and more popular for users and will account for an extremely high ratio of future Internet traffic.
However, diverse requirements of quality of service (QoS) guarantees from those video services bring new challenges to the network operation \cite{8016573}.
Video service enhancement has consequently become an important topic in both academia and industria \cite{MSheng,RKali}.

A fundamental problem for video service enhancement is how to send videos to the users as soon as possible.
Typically, there are two ways to deal with such problem, which are edge caching and video delivery optimization \cite{Gzhou,JWu}.
On the one hand, edge caching has been proposed to cache popular contents at the network edge, which shortens the distance between the request and service sources.
Due to the finite caching capacity and great number of user requests, how to cache videos from the cloud server is a key issue for edge caching study.
In order to improve quality of user experience (QoE) and reduce operation cost, it is expected as many users as possible hit their requests and as many cached videos as possible are hit at the edge server at the same time.
On the other hand, a reasonable video delivery scheme is able to improve the QoS and QoE requirements for various users.
Due to the limitation of communication resources, how to schedule users and manage resources to meet diverse requirements of video services is a continuous challenge in network optimization.
In order to guarantee the user fairness as well as the QoS and QoE requirements, the network is expected to provide more granular resource management for each video stream.

In the literature, video caching and video delivery are usually studied separately.
On video caching study, data-driven approaches (e.g., machine learning) are favorite due to their advantages in user popularity and user interest mining \cite{8658196,9165221,8842609,9148740,9145039,8352848}.
On video delivery study, model-driven methods are more popular, since they are applicable to deal with various of resource optimization problems \cite{CFAN,XZHANG,MCHOI,AAKH,XZHANG1}.
However, how to integrate video caching and delivery together to further enhance video service performance remains unknown for a longtime.
Recently, such challenge has attracts lots of attention from academia and some insightful schemes has been proposed with a prior knowledge of user preference.
Nevertheless, a systematical understand on jointly optimizing the video caching and delivery performance through user interest prediction and resource allocation is still unavailable.

Motivated by this, we develop a edge-cloud collaboration framework to enhance video service performance, where the cloud server decides which videos should be cached, and the edge server decides resource allocations for each user.
A data-driven hybrid human-artificial intelligence approach is first proposed to mine individual user interest and then generate group interest at the cloud server, based on which, videos can be selected to cache.
A statistical delay guarantee model is then derived to reveal the relationship between the video coding rate and delay constraint.
Thereafter, we propose a double bisection exploration scheme to help the edge server to find out optimal bandwidth allocation for each video stream, which ensures user fairness in terms of video coding rate.
Finally, extensive experiments are carried out to verify the effectiveness of the proposed video service enhancement scheme with the help a real data-set.

The contributions of this paper are summarized as follows:

\begin{itemize}
\item A factorization machine and multi-layer perceptron merging scheme is proposed to predict individual user interest.
As FM and MLP can well represent low-order and high-order features respectively.
The proposed scheme is able to guarantee high prediction accuracy on user interest of different videos.
\item A social aware similarity model is proposed to characterize the similarity between individual user and the group.
In addition, a group interest model is constructed with consideration of the impacts of user similarity, positive emotion and negative emotion.
The proposed group interest model is able to guarantee high user hit rate and high content hit rate at the same time.
\item A statistical delay guarantee model is derived analytically, based on which a double bisection exploration scheme is further proposed to guide video delivery.
The computation complexity of the proposed scheme is proved to be logarithmic level.
Hence, it is applicable to the realistic networks.
\end{itemize}

The rest of the paper is organized as follows.
The related works are reviewed in Section II.
In Section III, system model is introduced and optimization problem is formulated.
Video caching scheme is studied in Section IV while video delivery scheme is investigate in Section V.
In Section VI, experiments are carried out and results are presented and discussed.
Finally, the conclusion is summarized in Section VII.

\section{Related Work}
In this section, related works about service enhancement can be classified into content caching and content delivery.
Usually, content caching and content delivery are studied separately.

In the literature, content caching schemes were usually designed based on the content popularity \cite{8658196}.
In \cite{9165221}, a probabilistic dynamical model was proposed for content popularity prediction in terms of temporal and spatial dependencies.
In \cite{8842609}, the dynamic caching process in the vector space was studied with state transition field, under the assumption of time-invariant content popularity.
The caching schemes proposed by \cite{8658196,9165221,8842609} mainly focused on the content characteristics while user interest was not considered.
However, it has been verified that the content request is more sensitive to the user interest \cite{8627946}.
Hence, some works recently have concentrated on user interest mining while designing caching scheme.
Particularly, the idea of recommendation system have been confirmed to be powerful in individual interest predication \cite{article,8523627}.
In \cite{9148740}, the caching and recommendation decisions were jointly optimized in the context of ``soft cache hit'' setup.
In \cite{9145039}, the joint caching and recommendation process was modeled as a single caching policy in a fog computing network, which significantly reduced the interest training complexity.
In \cite{8352848}, an optimization problem for the joint caching and recommendation decisions was formulated with object to maximize the caching hit rate under minimal controllable distortion of the inherent user content preferences, where a heuristic algorithm was proposed to realize lightweight control over recommendations.
The caching schemes designed by works \cite{9148740,9145039,8352848} were based on individual user interest, which requires large caching capacity to ensure the user hit rate.
However, as mentioned before, edge caching is expected to satisfy the requests of multiple users with limited caching capacity.
How to cache contents based on group interest is still an open problem.

Related works about video delivery usually focused on resource allocation and network performance optimization.
In \cite{CFAN}, successful content delivery probability and energy efficiency were analyzed based on stochastic geometry theory under an edge caching framework, where the most popular contents are cached in the macro BSs tier with and the less popular contents are cached in the helpers tier.
In \cite{XZHANG}, a near-optimal video layer placement scheme was studied to maximize the total amount of data traffic that can be requested from the local small cells.
In \cite{MCHOI}, the focus was on a wireless caching network, where Markov decision process was applied to analyze the dynamic decision making process for video quality and chunk amounts, and Lyapunov optimization was used to decision the caching node.
Thereafter, a video delivery strategy was proposed to maximize time-average streaming quality under a playback delay constraint in wireless caching networks.
In \cite{AAKH}, the video delivery delay was modeled from a probabilistic point of view.
Effective capacity theory was introduced to optimize the resource allocation and user scheduling in terms of sum of video quality of each user.
In \cite{XZHANG1}, the authors proposed an information-centric virtualization architecture over 5G multimedia big data wireless networks, where information-centric network, network functions virtualization and software-defined networks were integrated to guarantee the statistical delay-bounded QoS for multimedia big data transmissions.

Works \cite{8658196,9165221,8842609,9148740,9145039,8352848,CFAN,XZHANG,MCHOI,AAKH,XZHANG1} made great efforts on service enhancement throughput either content caching or content delivery.
However, how to combine content caching and delivery together is unavailable from those works.
In order to further improve service performance, researchers recently pay attention on joint content caching and delivery.
In \cite{YJIANG}, a dynamic distributed edge caching scheme was constructed in ultra-dense F-RANs.
The request service delay and fronthaul traffic load were modeled into a cost function that was further minimized.
In \cite{HWU1}, the authors focused on a heterogeneous vehicular network where WiFi roadside units, TV white space stations, and cellular base stations were coexist to cache and deliver contents.
The stable-matching-based caching scheme was proposed to minimize the average delivery delay by jointly considering file characteristics and network conditions.
In \cite{SZHANG}, the focus was on cooperative edge caching in large-scale user-centric mobile networks, where video delivery delay was minimized through optimizing content placement and cluster size in terms of traffic distribution, channel quality and file popularity.
In \cite{LPU1}, a joint online resource allocation, content caching and request routing scheme was studied to minimize the system cost in a cloud radio access network, where storage, VM reconfiguration, latency, and content migration were all taken into account.
In works \cite{YJIANG,HWU1,SZHANG,LPU1}, a priori content popularity or user preference was assumed to be available, such that the content caching and delivery could be jointly optimized.
In practice, user interest, however, belongs to subjective emotion that is difficult to predict.
In \cite{ZZhang1}, the non-negative matrix factorization technique was applied to predict user's preference, based on which a novel hierarchical proactive caching scheme was proposed for individual user.
Nevertheless, the user preference model therein was only applied to single user scenario.

In summary, there is still lack of systematical understand on how to jointly enhance video caching and delivery performance through user interest prediction and wireless resource allocations for a multi-user network, which motivates this paper.

\section{System Model}
\subsection{Network Scenario}
\begin{figure}[t!]
\centerline{\includegraphics[scale=0.33]{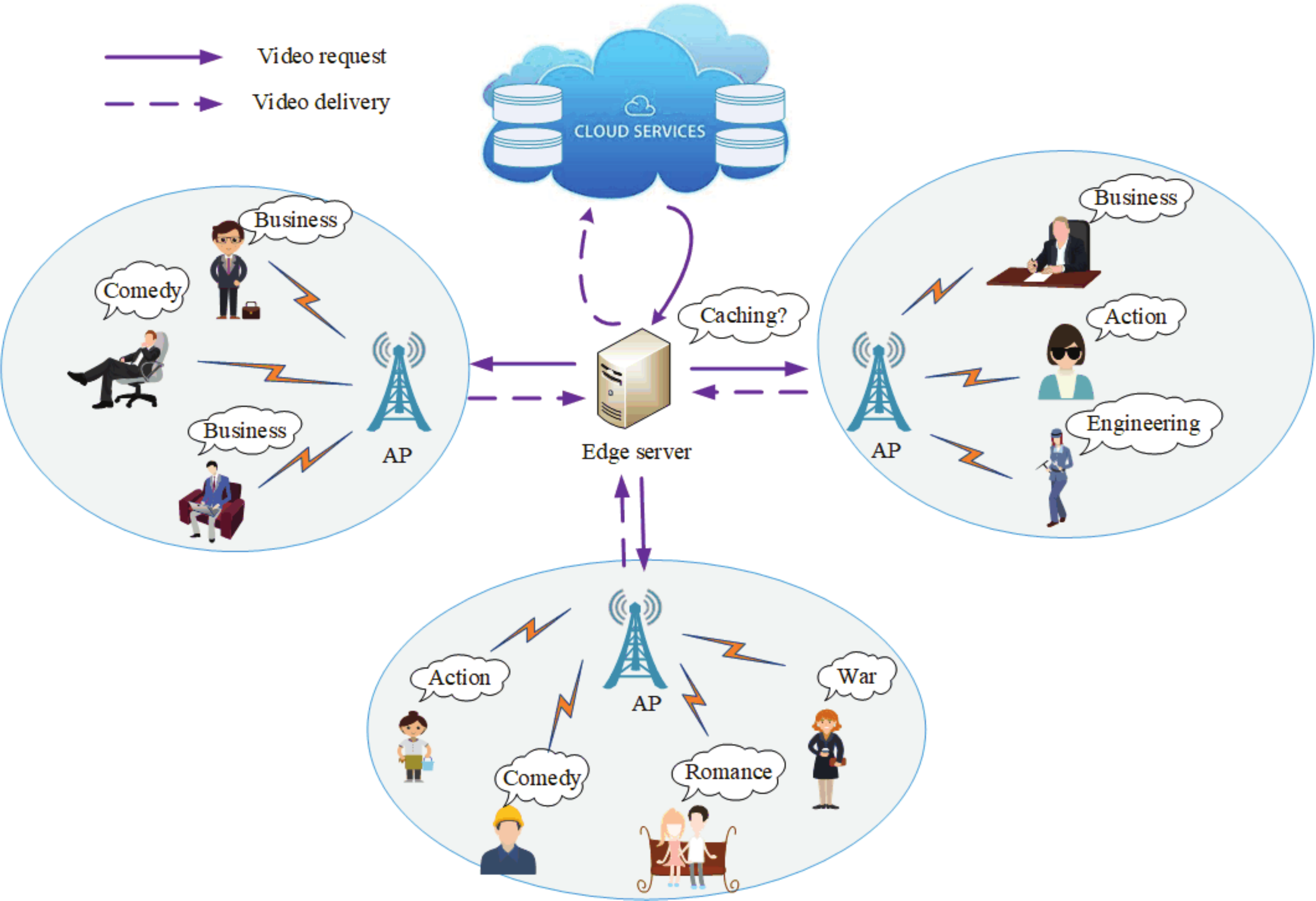}}
\caption{The considered edge caching network}
\label{fig1}
\end{figure}

The considered edge-cloud collaborative video caching and delivery framework is depicted in Fig. \ref{fig1}, where there are cloud server, edge server and multiple APs.
The cloud server stores all the videos that may be requested by users.
The videos stored at the cloud server are denoted by set ${\bf{\mathcal{M}}}=\{f_1,f_2,...,f_m,...,f_M\}$ with size $M$.
The cloud server is responsible for predicting the physical position and video preference of each user based on the information of user historical behaviors.
As a result, the information of user requests from each small cell can be predicted, based on which the cloud server determines which videos can be cached by each edge server periodically.
The edge server is deployed to cache videos for users in several small cells.
We use ${\bf{\mathcal{E}}}$ to denote the videos cached at the edge server and $E$ (in unit of videos) to denote the caching capacity of the edge server.
In each small cell, a user can send a video request to the edge server through the nearest AP.
If the requested video is requested by the edge server, it will be delivered to the user directly.
Otherwise, the edge server will first request the video from the cloud server and then transmit it to the user, which leads to an extra delay.

In addition to video caching, the edge server is also responsible for radio resource allocations for each downlink of video transmission from an AP to a user.
The users within the service coverage of the edge server are denoted by set $\mathcal{N}=\{U_1,U_2,...,U_n,...,U_N\}$ with size $N$.
The downlink channels allocated to the users are assumed to be orthogonal.
The total bandwidth of the considered network is denoted by $B$ and the bandwidth allocated to the $n$th user is denoted by $B_n$.
According to Shannon's theorem, the channel capacity of the $n$th user holds as
\begin{equation} \label{channel capacity}
\begin{aligned}
R_n(i)=B_n\log_2(1+\frac{p_nh_n(i)l_n(i)}{N_0B_n})
\end{aligned},
\end{equation}
where $p_n$ denotes the downlink transmission power, $h_n(i)$ denotes the channel gain due to small-scale fading, $l_n(i)$ denotes the channel gain due to large-scale fading, and $N_0$ denotes the power spectral density of white Gaussian noise.
Besides, $i$ is used to index the time block (TB) during a video rate adaption period, where the length of one TB is equal to the channel coherence time $T$.
Typically, the coherence time is in the level of milliseconds while the video rate adaptation is usually carried out in the seconds timescale.
Hence, the timescale of video rate adaptation is much larger than the channel coherence time.
It is assumed that the channel is quasi-static flat block fading.
In other words, $h_n$ is invariant within a TB but independently and identically distributed (i.i.d) among different TBs.
In addition, the position of a user is usually static when he or she is watching a video, which implies the time scale of channel varying caused by large-scale fading is not smaller than that of video rate adaptation.
Therefore, $l_n$ can be considered as constant among different TBs.

\subsection{Delay Model}
The total video delivery delay in the $i$th TB is denoted by $D_n(i)$.
According to Fig. \ref{fig1}, the video delivery delay includes three parts, i.e., the transmission delay from the cloud server to the edge server $D_n^{CE}(i)$, that from the edge server to the AP $D_n^{EA}(i)$, and that from the AP to the $n$th user $D_n^{AU}(i)$, where the first two parts belong to wire communications while the third part belong to wireless communications.
We have
\begin{equation} \label{total delay}
\begin{aligned}
D_n(i)=D_n^{CE}(i)(1-\sum_{f_m \in \mathcal{M}}\alpha_{n,m}\beta_m)+D_n^{EA}(i)+D_n^{AU}(i)
\end{aligned},
\end{equation}
where $\alpha_{n,m}=1$ means video $f_m$ is requested by the $n$th user, and $\alpha_{n,m}=0$ otherwise.
Besides, $\beta_m=1$ means video $f_m$ is cached by the edge server, i.e., $f_m \in \mathcal{E}$, and $\beta_m=0$ otherwise.

For the wire communication part, the transmission delay can be model by a constant due to the reason that the wire bandwidth allocated to a video stream is fixed during a rate adaption period.
Furthermore, as the distance between the edge server and the APs is much shorter than that between the cloud server and the edge server, the transmission delay from the edge server to the AP can be neglected.
Hence, we can let $D_n^{CE}(i)=d^{C}$ and $D_n^{EA}(i)\approx 0$, there holds
\begin{equation} \label{total delay2}
\begin{aligned}
D_n(i)=d_n^{C}(1-\sum_{f_m \in \mathcal{M}}\alpha_{n,m}\beta_m)+D_n^{AU}(i)
\end{aligned}.
\end{equation}

On the other hand, as wireless channel is time-varying, it is infeasible to provide deterministic delay guarantee for users.
A queue is therefore inserted at the edge server for each user stream to absorb the mismatch between the arrival and service rate due to the channel variations.
In the subsequence, time interval $[j,i)$ is used to represent the time from the $j$th TB to the $i$th TB.
The cumulative amount of stream arrivals for the queue corresponding to the $n$th user is denoted by $A_n(j,i)=V_n(i-j)T$ during $[j,i)$, where $V_n$ denotes the video coding rate during a video rate adaption period.
Similarly, the corresponding cumulative amount of departures is denoted by $A_n^*(j,i)$.
It is easily verified that for a queue with input $A_n(j,i)$ and output $A_n^*(j,i)$, there holds \cite{SNetCal}
\begin{equation} \label{minplus}
\begin{aligned}
A_n^*(0,i) = \inf_{0\leq j \leq i } \{A_n(0,j)+C_n(j,i)\}
\end{aligned},
\end{equation}
where $C_n(j,i)=\sum_{t=j}^{i-1}R_n(t)T$ denotes the cumulative amount of video steam that the AP can transmit during $[j,i)$.
For the video stream on the wireless channel, the transmission delay from the AP to the $n$th user $D_n^{AU}(i)$ can be modeled as follows
\begin{equation} \label{wiredelay}
\begin{aligned}
D_n^{AU}(i)=\inf\{d_n^{AU}:A_n(0,i)\leq A_n^*(0,i+d_n^{AU})\}
\end{aligned},
\end{equation}
which means that the last bit of video stream arriving at the queue in the $i$th block is transmitted in the $(i+d_n^{AU})$th block.

In order to characterize the delay performance more intuitively, we model the delay metric according to the philosophy behind 5G ultra reliable low latency communications (uRLLC).
In specific, the delay constraint is modeled from a statistical point of view, there holds
\begin{equation} \label{delay constraint}
\begin{aligned}
\text{DVP}_n (d_n) \triangleq \Pr\{D_n(t)>d_n\}\leq \epsilon_n
\end{aligned},
\end{equation}
meaning that the probability that video delivery delay exceeding the maximum tolerance $d_n$ should be control within $\epsilon_n$.

\subsection{Problem Formulation}
Intuitively, the video delay performance for a user can be improved from the following the video caching and video delivery aspects.
On one hand, the video delivery delay will be reduced if the user request can be predicted with a high accuracy, since the transmission from the cloud server to the edge server can be avoided.
On the other hand, larger channel capacity, i.e., more allocated resources, can also guarantee a lower transmission delay.
Additionally, the delay performance can be improved by lowering video coding rate.
However, lower video coding rate results in worse quality of user experience (QoE).
In order to guarantee the fairness of QoE among users, a optimization problem is formulated to maximize the minimum individual video coding rate for the considered network.
\begin{equation} \label{p1}
\begin{aligned}
\text{\textbf{P1}}~~~~~& {\max\limits_{\mathbf{V},\mathbf{B},\bf{\beta}} }~~~~{\mathop {\min\limits_{U_n \in \mathcal{N}} }}~{V_n}\\
\text{s.t.}~~~~~&\text{C1}:~{\text{DVP}_n (d_n)\leq \epsilon_n,~\forall U_n \in \mathcal{N}}\\
&\text{C2}:~{\sum\limits_{U_n \in \mathcal{N}}B_n\leq B}\\
&\text{C3}:~{\sum\limits_{f_m \in \mathcal{M}}\alpha_{n,m}\leq 1,~\forall U_n \in \mathcal{N}}\\
&\text{C4}:~{\sum\limits_{f_m \in \mathcal{M}}\beta_{m}\leq E}\\
\end{aligned},
\end{equation}
where $\mathbf{V}=\{V_n:U_n \in \mathcal{N}\}$ denotes the video coding rate vector, $\mathbf{B}=\{B_n:U_n \in \mathcal{N}\}$ denotes the bandwidth allocation vector, and $\mathbf{\beta}=\{\beta_m:f_m \in \mathcal{M}\}$ denotes the video caching vector.
In problem P1, constraint C1 means the statistical delay requirement of each user, C2 means that the resource allocated to the users should not be greater than the total amount in the considered network, C3 represent that each user can only request one video at a time, C4 means that the number of cached video should not exceed the caching capacity of the edge sever.

Note that C1-C2 belong to video delivery constraints while C3-C4 belong to video caching constraints.
Besides the video caching process and the video delivery process are independent with each other.
Hence, problem P1 can be decoupled to a video caching subproblem and a video delivery subproblem.
Typically, a higher user hit rate in the edge caching can guarantee better video delay performance, which further sustains higher video coding rate.
The video caching subproblem regarding to user hit rate is formulated as follows
\begin{equation} \label{p2}
\begin{aligned}
\text{\textbf{P2}}~~~~~& \max\limits_{\bf{\beta}} \frac{\sum\limits_{U_n \in \mathcal{N}}\sum\limits_{f_m \in \mathcal{M}}\alpha_{n,m}\beta_m}{\sum\limits_{U_n \in \mathcal{N}}\sum\limits_{f_m \in \mathcal{M}}\alpha_{n,m}}\\
\text{s.t.}~~~~~&\text{C3}:~{\sum\limits_{f_m \in \mathcal{M}}\alpha_{n,m}\leq 1,~\forall U_n \in \mathcal{N}}\\
&\text{C4}:~{\sum\limits_{f_m \in \mathcal{M}}\beta_{m}\leq E}\\
\end{aligned}
\end{equation}
Since video coding rate is directly correlated with the video delivery performance, the video delivery subproblem regarding to resource allocation is formulated as follows
\begin{equation} \label{p3}
\begin{aligned}
\text{\textbf{P3}}~~~~~& {\max\limits_{\mathbf{V},\mathbf{B}} }~~~~{\mathop {\min\limits_{U_n \in \mathcal{N}} }}~{V_n}\\
\text{s.t.}~~~~~&\text{C1}:~{\text{DVP}_n (d_n)\leq \epsilon_n,~\forall U_n \in \mathcal{N}}\\
&\text{C2}:~{\sum\limits_{U_n \in \mathcal{N}}B_n\leq B}\\
\end{aligned}
\end{equation}
The solution of problems P2 and P3 will be introduced in the following two sections.

\section{Video Caching Scheme}
In this section, a hybrid human-artificial intelligence caching scheme is studied to improve the user hit rate and content hit rate for the edge server.
Compared with the objective factors, the caching performance is more sensitive to subjective factors, such as user interest, user relationship and etc.
The impact of subjective factors on user behavior prediction is infeasible to represent analytically.
Hence, it is difficult or impossible to find out the optimal solution for problem P2.
In this paper, we propose to design a video caching scheme with data-driven approach, which can also be considered as a sub-optimal solution for problem P2.
In specific, we first propose a machine learning framework to mine the user interest on videos.
Considering the caching constraint, a social computing idea is learned to studied the similarity between each user and the group users.
Based on the obtained similarity model, the videos can be selected to cache at the edge server.

\begin{figure*}[t!]
\centerline{\includegraphics[height=3in,width=4in]{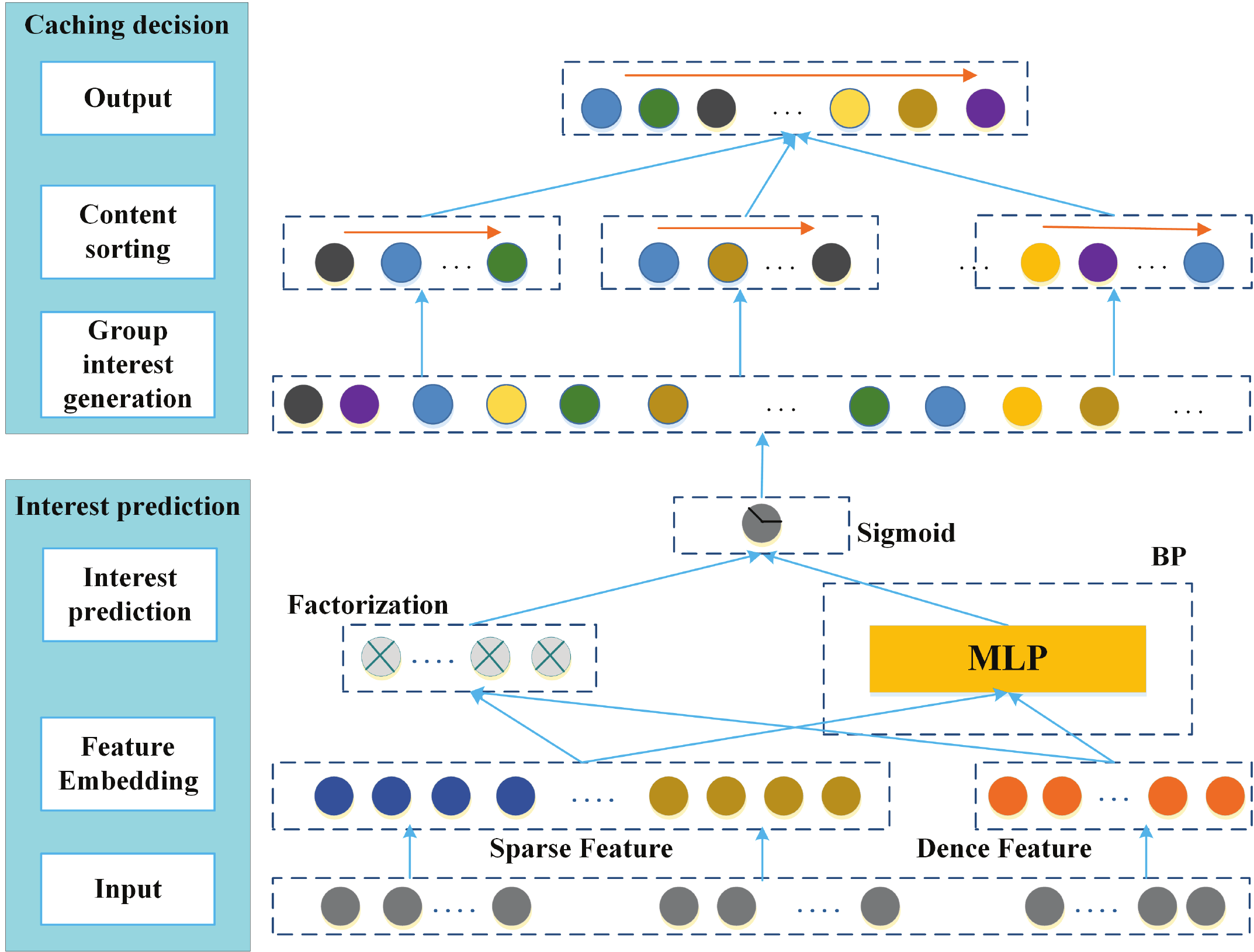}}
\caption{Edge-cloud collaborative video caching and delivery framework}
\label{fig2}
\end{figure*}

In practical network, in order to reduce the traffic load for the backhaul and backbone network, the timescale of video caching duration should be usually much greater than that of video rate adaption.
Therefore, each user may request various videos while each video may be requested by different users during a caching decision duration.
As a caching decision is highly correlated with user interest, we focus on interest mining with the help of user attributes, video characteristics and historical user behaviors.
The overall idea of our caching scheme can be summarized as follows,

\begin{itemize}
\item Step-1: User attribute modeling.
User attributes include age, occupation, residence and etc.
According to psychology, user attributes have great difference on user interest \cite{wang2011cognitive}.
Typically, user interest can be characterized more accurately if more types of user attributes are available.

\item Step-2: Video characteristic modeling.
Content characteristic include video type, the published year and etc.
Intuitively, video characterizes have impact on user's decision.
The more details can a video be characterized, the more accurate user interest prediction can be guaranteed.

\item Step-3: Interest prediction modeling.
In order to decide which videos to be cached, the interest of each user on each video can be analyzed.
As the user attributes and video characteristics include both low-order and high-lower features.
A machine learning based prediction model is needed to learn the individual user interest from the aforementioned features.

\item Step-4: Social-aware video caching.
To improve the user hit rate, the obtained individual user interest should be merged into group interest.
Specifically, social relationship between each user and the considered group should be taken into account, since a user with stronger relationship may reflect the group interest more accurately.
According to the group interest on the candidate videos from the cloud server, edge caching is finally completed.
\end{itemize}

In summary, we construct a hybrid human-artificial intelligence caching framework (HHAICF) as depicted in \ref{fig2}.
The HHAICF includes the interest prediction part and the caching decision part.
The interest prediction part is responsible for predicting the user preference on each video.
This part consists of the input layer, the feature embedding layer and interest prediction layer.
The caching decision part is responsible for caching contents based on group interest.
This part consists of group interest generation layer, content sorting layer and output layer.

\subsection{Machine Learning Based Individual Interest Prediction}
Typically, a user's interest on a given video can be classified as like or unlike, which can be modeled as a binary classification problem.
It has been proved that FM can realize the low-order features crossing while MLP can learn from the high-order features well \cite{cheng2016wide,DBLP:journals/corr/GuoTYLH17}.
Hence, we propose to combine FM with MLP, such that the complex features can be learned and the simple features can be retained at the same time.

Firstly, as labeled features have no practical meaning, we convert them into sparse features to reduce the data variability.
In contrast, the features with practical significance (e.g., year) should be converted into dense features for better representation.
For instance, the input feature ${\bf{x}}= \{1,2,3,1991,1992,1995\}$ can be divided into labeled feature ${{\bf{x}}_1}$ and feature with practical significance ${{\bf{x}}_2}$, where
\begin{equation}
{\rm{{{\bf{x}}_1} = \{1,2,3\}}},
{\rm{{{\bf{x}}_2} = \{1991,1992,1995\}}}.
\end{equation}
We can then embed ${{\bf{x}}_1}$ as a sparse feature with one-hot encoding, there holds
\begin{equation}
{{\bf{x}}_1} = \{001,010,100\}.
\end{equation}
Therefore, the labeled feature turns to be machine-readable with representation.

For feature with practical significance ${{\bf{x}}_2}$, it is embedded into a dense feature through normalization as follows
\begin{equation}
{{\bf{x}}_2} = \frac{{{{\bf{x}}_2} - \mu }}{\sigma}.
\end{equation}
where $\mu$ denotes the mean of ${\bf{x}}_2$, and $\sigma$ denotes the variance.
The reason of normalization is to reduce the difference between the values for feature representation.
In order to reduce the training complexity, the embedded features are shared by the FM module and MLP module.

In the FM module, the first-order matric is denoted by $\bf{w}_\text{FM}$, the second-order weight matric is denoted by $\bf{W}_\text{FM}$.
Moreover, $\bf{W}_\text{FM}$ can be decomposed as
\begin{equation}
\bf{W}_\text{FM} = {\bf{Y}}{{\bf{Y}}^T},
\end{equation}
where
\begin{equation}
{\bf{Y}} = \left( {\begin{array}{*{20}{c}}
{{y_{11}}}&{{y_{12}}}& \ldots &{{y_{1N}}}\\
{{y_{21}}}&{{y_{22}}}& \ldots &{{y_{2N}}}\\
 \vdots & \vdots & \ddots & \vdots \\
{{y_{M1}}}&{{y_{M2}}}& \ldots &{{y_{MN}}}
\end{array}} \right) = \left( {\begin{array}{*{20}{c}}
{\bf{y}_1}\\
{\bf{y}_2}\\
 \vdots \\
{\bf{y}_M}
\end{array}} \right).
\end{equation}

The output of the FM module is obtained as
\begin{equation}\label{factorization machine}
\begin{array}{l}
{{\bf{Z}}_{\rm{FM}}} =  < \mathbf{w},\mathbf{x} >  + \sum\limits_{m = 1}^{M-1} {\sum\limits_{i = m + 1}^M { < {{\bf{y}}_m},{{\bf{y}}_i} > {x_m}} } {x_i}\\
 =  < \mathbf{w},\mathbf{x} >  + \frac{1}{2}\sum\limits_{n = 1}^{N} {({{(\sum\limits_{m = 1}^{M} {{y_{m,n}}{x_m}} )^2}} - \sum\limits_{m = 1}^M {y_{m,n}^2x_m^2} )}
\end{array},
\end{equation}
where $<\mathbf{w},\mathbf{x}>$ denotes the inner product operation between vectors $\mathbf{w}$ and $\mathbf{x}$.
In this process, the lower-order features are crossed in pairs, such that the implicit relationship between features can be mined.

In MLP module, the ReLU function is employed as activation function.
The output of MLP is expressed as
\begin{equation}\label{deep}
{{\bf{Z}}_{\mathop{\rm MLP}\nolimits} } = \text{ReLU}({\bf{W_\text{MLP}}}*{\bf{x}} + {\bf{b}}),
\end{equation}
where $\bf{W}_\text{MLP}$ denotes the weight matric of the MLP module, ${\bf{b}}$ denotes the bias vector.

Thereafter, the output of FM module ${{\bf{Z}}_{\rm{FM}}}$ and that of the MLP module ${{\bf{Z}}_{\rm{MLP}}}$  are concated and processed by the Sigmoid function.
The final optimal parameter $\bf{W}=\{\bf{w}_\text{FM}, \bf{W}_\text{FM},\bf{W}_\text{MLP}\}$ is then obtained by continuous gradient update.
Hence, the output of the interest prediction model holds as
\begin{equation}
{\bf{\widetilde{r}}} = \text{Sigmoid}({{\bf{Z}}_\text{FM}},{{\bf{Z}}_\text{MLP}}).
\label{giacf}
\end{equation}
Note that ${\bf{\widetilde{r}}}$ is a vector where the elements therein are all valued from 0 to 1.
It quantifies the relevance between the users and videos.
To optimize the learning weight ${\bf{W}}$, we use the cross-entropy function as the loss function, i.e.,
\begin{equation}\label{loss}
\text{Loss} =  - \frac{1}{||\mathbf{x}||}\sum\limits_{x \in \mathbf{x}} {(r\log {{\widetilde{r}}}(x) + (1 - r)\log (1 - {{\widetilde{r}}}(x)))} ,
\end{equation}
where $r \in \{ 0,1\} $ denotes the label of training set, and ${{\widetilde{r}}}(x)$ denotes the interest prediction corresponding to input $x$, $||\mathbf{x}||$ denotes the size of the training set.

The combination of FM and MLP ensures both low-order and high-order features can be represented, such that the individual user interest can be captured.
Moreover, based on individual user interest, group interest can be further represented to optimize the content caching.

\subsection{Social-Aware Group Interest Based Video Caching}
With the obtained interest prediction model, the individual user interest on the new videos can be predicted as
\begin{equation}\label{newinterest}
{\bf{\widetilde{r}}}^{new} = {\bf{\widetilde{r}}}(\mathbf{x}^{new}),
\end{equation}
where $\mathbf{x}^{new}$ denotes the new input corresponding to the new videos that never be watched by the considered users, ${\bf{\widetilde{r}}}^{new}$ denotes the set of prediction results.

In order to guarantee a high user hit rate, it is necessary to mine the group interest from the individual user interest on the new videos.
In this paper, we resort to the idea of social computing to model the group interest.
In specific, two user with similar history behavior of video watching probably have similar interest on videos.
The interest similarity between $U_{n_1}$ and $U_{n_2}$ is usually characterized by cosine function, i.e.,
\begin{equation} \label{cosin}
\text{simc}_{n_1,n_2}=\frac{\mathcal{X}(n_1)\cup \mathcal{X}(n_2)}{\sqrt{|\mathcal{X}(n_1)||\mathcal{X}(n_2)|}},
\end{equation}
where $\mathcal{X}(n_1)$ denotes the video set that $U_{n_1}$ have watched before, and $|\mathcal{X}(n_1)|$ denotes the size of set $\mathcal{X}(n_1)$.
Generally, some users may watch some popular videos before though their interest on those videos are different.
In order to avoid misleading on interest similarity estimation by popular videos, we adjust the classical cosine similarity model by introducing a penalty factor, there holds
\begin{equation} \label{simi}
\text{simi}_{n_1,n_2}=\frac{\sum\limits_{\mathcal{X}(n_1)\cup \mathcal{X}(n_2)}\frac{1}{\ln(1+|\mathcal{I}(m)|)}}{\sqrt{|\mathcal{X}(n_1)||\mathcal{X}(n_2)|}}.
\end{equation}
Here $\mathcal{I}(m)$ denotes the user set that have watched the video $f_m$ and $|\mathcal{I}(m)|$ denote the corresponding size.
The more popular $f_m$ is, the greater $|\mathcal{I}(m)|$ holds.
Furthermore, the similarity between the $U_n$ and the group can be modeled as the sum of the similarity between $U_{n_1}$ and each user, which holds as
\begin{equation} \label{groupsi}
\text{simg}_{n}=\sum_{U_{n_2}\in \mathcal{N} \&\&n_2\neq n}\frac{\sum\limits_{\mathcal{X}(n)\cup \mathcal{X}(n_2)}\frac{1}{\ln(1+|\mathcal{I}(m)|)}}{\sqrt{|\mathcal{X}(n)||\mathcal{X}(n_2)|}}.
\end{equation}
A user with higher similarity with the group means that the interest of this user on videos can better represent the group interest.
Hence, we further model the impact of similarity between the user and the group by the following expression,
\begin{equation} \label{weight}
a_n^{si}=\frac{\text{simg}_{n}-\min\{\text{simg}_{n}:U_n \in \mathcal{N}\}}{\max\{\text{simg}_{n}:U_n \in \mathcal{N}\}-\min\{\text{simg}_{n}:U_n \in \mathcal{N}\}}
\end{equation}

Note that the interest prediction for a user on a video can be classified into positive case or negative case.
Let $\widetilde{r}^{new}_{n,m} \in {\bf{\widetilde{r}}}^{new}$ denotes the preference of $U_n$ on the $m$th video.
Besides, we use $\delta$ to denote the threshold, where $\widetilde{r}^{new}_{n,m}>\delta$ means that the $U_n$ represents positive emotion to the $m$th video, otherwise the $U_n$ represents negative emotions to the $m$th video.
Accordingly, the proportion of users that represent positive emotion to the $m$th video and that represent negative emotion holds respectively
\begin{equation} \label{proportion}
\begin{aligned}
a_{m}^{po}=\frac{\sum\limits_{U_n \in \mathcal{N}}\mathbf{I}_{\{\widetilde{r}^{new}_{n,m}\geq \delta\}}}{N}\\
a_{m}^{ne}=\frac{\sum\limits_{U_n \in \mathcal{N}}\mathbf{I}_{\{\widetilde{r}^{new}_{n,m}< \delta\}}}{N}\\
\end{aligned},
\end{equation}
where $\mathbf{I}_{\text{event}}$ is an indicator function.
Specifically, $\mathbf{I}_{\text{event}}=1$ if event is true, otherwise $\mathbf{I}_{\text{event}}=0$.

Based on the obtained individual interest on the new videos, we model the group interest on the $m$th new video by the following expression
\begin{equation} \label{score}
\text{Pre}_m= \frac{\sum\limits_{U_n \in \mathcal{N}}(a_m^{po}a_n^{si}\widetilde{r}^{new}_{n,m}\mathbf{I}_{\{\widetilde{r}^{new}_{n,m}\geq \delta\}}+a_m^{ne}a_n^{si}\widetilde{r}^{new}_{n,m}\mathbf{I}_{\{\widetilde{r}^{new}_{n,m}\leq \delta\}})}{N}.
\end{equation}
The intuitions behind the group interest model are as follows.
Firstly, $a_n^{si}$ characterizes the impact of the similarity between $U_n$ and the group on the group interest prediction.
The higher similarity between $U_n$ and the group, the heavier weight should be allocated to $U_n$ to ensure video prediction meet the interest of most of the users in the group better.
Secondly, $a_{m}^{po}$ is introduced to enhance the impact of positive case on the group prediction.
With this regard, we can make sure the video that is preferred by most of the users in the group to be a high prediction score.
Thirdly, $a_m^{ne}$ is used to improve the selection possibility of the video that is strongly preferred by some users but may not disgust other users.

According to (\ref{score}), the group interest on the candidate videos can be obtained and sorted in the descending order.
For an edge server with caching capacity of $E$ ($C\ll M$) videos, the videos of with top $E$ prediction scores can be cached.
In other words, the caching decision can be obtained as
\begin{equation}\label{cade}
\beta_m=
\begin{cases}
1, ~~\text{Pre}_m~ \text{is ~within ~the ~top}~E~\text{prediction~scores}\\
0, ~~\text{otherwise}
\end{cases}.
\end{equation}

The set of contents that are cached is finally obtained as
\begin{equation}\label{finalcache}
\mathcal{E} =\{f_m:\beta_m==1,~m=1,2,...,E\}.
\end{equation}
The hybrid human-artificial intelligence caching scheme is summarized as Algorithm 1.

\begin{algorithm}
  \caption{Hybrid human-artificial intelligence caching scheme}
  \label{alg:Framwork2}
  \begin{algorithmic}[1]
  \Require
      Interaction information between users and contents $\mathbf{x}$, caching capacity of edge server $E$;
    \Ensure
    Video caching decision $\beta$ and videos selected to cache $\mathcal{E}$;
    \State  Data set partition and obtain the training set $\mathbf{x}^{tr}$, the test set $\mathbf{x}^{te}$, and the set of new contents $\mathbf{x}^{new}$;
    \State \textbf{for} epoch=$1,2,...,e,...$ \textbf{do}
        \State \quad Train weights with FM, and obtain user interest ${{\bf{Z}}_{\rm{FM}}} $ according to (\ref{factorization machine});
        \State \quad  Train weights with MLP, and obtain user interest ${{\bf{Z}}_{\rm{MLP}}} $ according to (\ref{deep});
        \State \quad  Predict user interest by merging the results from MLP and FM according to (\ref{giacf});
        \State \quad Compute the loss according to (\ref{loss});
        \State \quad \textbf{if} the $e$-th Loss \textless the $(e-1)$-th Loss \textbf{do}
            \State \quad \quad \textbf{Continue};
        \State \quad \textbf{else}  \textbf{do}
            \State \quad \quad Obtain optimal weight ${\bf{W}}$ and break;
        \State \quad \textbf{end if}
    \State \textbf{end for}
    \State Calculate user interest for $\mathbf{x}^{new}$ according to (\ref{newinterest});
    \State Calculate the similarity between each user and the group according to (\ref{groupsi});
    \State Calculate weights for group interest according to (\ref{weight}) and (\ref{proportion});
    \State \textbf{for} $f_m \in \mathcal{M}$ \textbf{do}
        \State \quad Calculate the group interest for the $m$th new video according to (\ref{score}).
    \State \textbf{end for}
    \State Obtain $\beta$ from (\ref{cade}) and $\mathcal{E}$ from (\ref{finalcache});
  \end{algorithmic}
\end{algorithm}

\section{Video Delivery Scheme}
In this section, the video delivery scheme is designed based on the edge caching result $\beta$ and the user request $\{\alpha_{n,m}\}$.
To optimize the wireless resource and video coding rate in problem P3, we first derive the delay violation probability, which is summarized as the following theorem.
\begin{thm}
For a user requesting video with delay requirement $d_n$, if video coding rate and channel capacity are set as $V_n$ and $R_n(i)$ respectively, the corresponding delay violation probability is upper bounded by
\begin{equation}\label{devio}
\begin{aligned}
\text{DVP}_n(d_n)&=\Pr\{D_n(i)>d_n\}\\
&\leq \mathbb{E}[e^{-\theta_nR_n(i)T}]^{(d_n-d_n^{C}(1-\sum\limits_{f_m \in \mathcal{M}}\alpha_{n,m}\beta_m))}
\end{aligned}
\end{equation}
for any $\theta_n$ satisfying
\begin{equation}\label{stability}
\begin{aligned}
V_n \leq -\frac{\ln\mathbb{E}[e^{-\theta_n R_n(i)T}]}{\theta_n T}
\end{aligned},
\end{equation}
where $\mathbb{E}[\cdot]$ represents the expectation operator.
\end{thm}
\begin{proof}
According to (\ref{total delay2}), event $\{D_n(i)>d_n\}$ is equivalent to event $\{D_n^{AU}(i)>d_n-d_n^{C}(1-\sum\limits_{f_m \in \mathcal{M}}\alpha_{n,m}\beta_m)\}$.
Further according to the delay model (\ref{wiredelay}), event $\{D_n^{AU}(i)>d_n-d_n^{C}(1-\sum\limits_{f_m \in \mathcal{M}}\alpha_{n,m}\beta_m)\}$ implies event $\{A_n(0,i)>A_n^*(0,i+d_n-d_n^{C}(1-\sum\limits_{f_m \in \mathcal{M}}\alpha_{n,m}\beta_m))\}$, we consequently have $\{D_n^{AU}(i)>d_n-d_n^{C}(1-\sum\limits_{f_m \in \mathcal{M}}\alpha_{n,m}\beta_m)\} \subseteq \{A_n(0,i)>A_n^*(0,i+d_n-d_n^{C}(1-\sum\limits_{f_m \in \mathcal{M}}\alpha_{n,m}\beta_m))\}$.
Besides, as the randomness of channel capacity $R_n$ depends on the small-scale fading $h_n$ that is i.d.d, $R_n$ is also an i.i.d variable.
Moreover, the video coding process is independent of channel fading, the arrival process and service process of the user stream queue are independent.
Let $\Delta_n=(1-\sum\limits_{f_m \in \mathcal{M}}\alpha_{n,m}\beta_m))$, there holds
\begin{equation}
\begin{aligned} \notag
&\Pr\{D_n^{AU}(i)>d_n-d_n^{C}\Delta_n\}\\
\leq &\Pr\{A_n(0,i)-A_n^*(0,i+d_n-d_n^{C}\Delta_n>0\}\\
\overset{(a)}{=}&\Pr\{A_n(0,i)-\inf_{0\leq j\leq {i+d_n-d_n^{C}\Delta_n}}\{A_n(j) \\
&~~~~~~~~~~~~~~+C_n(j,i+d_n-d_n^{C}\Delta_n)\}>0\}\\
=&\Pr\{\sup_{0\leq j\leq i}\{A_n(j,i)-C_n(j,i)\}>C_n(0,d_n-d_n^{C}\Delta_n)\}\\
=&\Pr\{\sup_{0\leq j\leq i}\{e^{A_n(j,i)-C_n(j,i)}\}>e^{C_n(0,d_n-d_n^{C}\Delta_n)}\}\\
\overset{(b)}{\leq}&\mathbb{E}[e^{-\theta_n C_n(0,d_n-d_n^{C}\Delta_n)}]\mathbb{E}[e^{\theta_n A_n(0,1)}]\mathbb{E}[e^{-\theta_n C_n(0,1)}]\\
\overset{(c)}{=}&\mathbb{E}[e^{-\theta_n \sum_{i=1}^{d_n-d_n^{C}\Delta_n}R_n(i)T}]e^{\theta_n V_nT}\mathbb{E}[e^{-\theta_n R_n(i)T}]\\
\overset{(d)}{\leq}&\prod_{i=1}^{d_n-d_n^{C}\Delta_n}\mathbb{E}[ e^{-\theta_n R_n(i)T}]\\
{=}&\mathbb{E}[e^{-\theta_nR_n(i)T}]^{(d_n-d_n^{C}(1-\sum\limits_{f_m \in \mathcal{M}}\alpha_{n,m}\beta_m))}\\
\end{aligned}
\end{equation}
Here, step (a) is according to (\ref{minplus}).
In step (b), we apply the supermartingale theory and Chernoff bound \cite{lzd}.
Step (c) is according to the definitions of $A_n$ and $C_n$.
Besides, step (c) and step (d) hold due to the reason that $R_n$ is i.i.d, where $\mathbb{E}[e^{-\theta_nR_n(i)}]$ is always fixed for $\forall i\geq 0$.
Therefore, Theorem 1 is proved.
\end{proof}

Theorem 1 also reveals that the maximum video coding rate for a user holds as
\begin{equation} \label{Vn}
\begin{aligned}
V_n^{\max}=-\frac{\ln\mathbb{E}[e^{-\theta_n R_n(i)T}]}{\theta_n T},
\end{aligned}
\end{equation}
when other conditions are fixed.
In addition, $\theta_n$ can be regarded as a free parameter.
If the delay requirement and delay violation probability for a video stream are given, according to (\ref{devio}), there holds
\begin{equation}
\begin{aligned} \label{theta}
&{\mathbb{E}[e^{-\theta_nR_n(i)T}]^{(d_n-d_n^{C}(1-\sum\limits_{f_m \in \mathcal{M}}\alpha_{n,m}\beta_m))} = \epsilon_n}\\
&~~~~~~~~~~~~~~~~~~~~~~~~~~\Updownarrow\\
&\theta_n=-\frac{\ln\epsilon_n}{V_nT(d_n-d_n^{C}(1-\sum\limits_{f_m \in \mathcal{M}}\alpha_{n,m}\beta_m))}
\end{aligned}
\end{equation}
As a result, problem P3 can be transferred to the following problem.
\begin{equation} \label{p4}
\begin{aligned}
\text{\textbf{P4}}~~~~~& {\max\limits_{\mathbf{V},\mathbf{B}} }~~~~{\mathop {\min\limits_{U_n \in \mathcal{N}} }}~V_n=-\frac{\ln\mathbb{E}[e^{-\theta_n R_n(i)T}]}{\theta_n T}\\
\text{s.t.}~~~~~&\text{C1}:~\theta_n=-\frac{\ln\epsilon_n}{V_nT(d_n-d_n^{C}(1-\sum\limits_{f_m \in \mathcal{M}}\alpha_{n,m}\beta_m))}\\
&\text{C2}:~{\sum\limits_{U_n \in \mathcal{N}}B_n\leq B}\\
\end{aligned}
\end{equation}

To solve problem P4, a theorem is further proposed to illustrate how to balance the video coding rate and bandwidth allocation for each user.
\begin{thm}
The optimal solution of problem P4 always satisfies the following conditions
\begin{equation} \label{solcon}
\begin{aligned}
&V_1=V_2=\cdot\cdot\cdot V_n=\cdot\cdot\cdot=V_N\\
&\sum\limits_{U_n \in \mathcal{N}}B_n = B
\end{aligned}
\end{equation}
\end{thm}
\begin{proof}
Theorem 2 can be proved by contradiction approach.
Assume that when optimal solution $V_{n*}=\min\{V_n:U_n \in \mathcal{N}\}$ is obtained, there still exist two users $U_{n_1}$ with video coding rate $V_{n_1}>V_{n*}$.
Let $B_{n_1}$ and $B_{n*}$ denote the corresponding allocated bandwidth.
Now we reallocate the bandwidth to $U_{n_1}$ and $U_{n*}$ with $\widetilde{B}_{n_1}=(\widetilde{B}_{n_1}-\Delta B)$ and $(B_{n*}+\Delta B)$ respectively, where $\Delta B\rightarrow 0^+$
The new maximum video coding rate can be archived as $\widetilde{V}_{n_1}$ and $\widetilde{V}_{n*}$ respectively according to (\ref{Vn}).
Furthermore, from (\ref{Vn}) and (\ref{channel capacity}), video coding rate $V_n$ is positively related to channel capacity $R_n$ which increases with bandwidth $B_n$.
Therefore, we have $V_{n_1}>\widetilde{V}_{n_1}>\widetilde{V}_{n*}>V_{n*}$, which brings to the contradiction of the assumption that $V_{n*}$ is the optimal solution.
Hence, when problem P4 is optimally solved, the video coding rate for each user must be identical.

Similarly, assume that when optimal solution is obtained as $\mathbf{B}^*$ and $V^*=V_1=\cdot\cdot\cdot=V_n$, there holds $\sum_{U_n \in \mathcal{N}}B_n^* =B^*<B$.
We construct another bandwidth allocation scheme as $\widetilde{B}_n=\frac{B-B*}{N},~\forall U_n \in \mathcal{N}$.
The new obtained video coding rate for each user is denoted by $\widetilde{V}_{n}$.
As video coding rate is an increasing function of bandwidth according to (\ref{Vn}) and (\ref{channel capacity}), we have $\widetilde{V}_{n}>V^*,~\forall U_n \in \mathcal{N}$, which brings to the contradiction of the assumption.
Thus, the network bandwidth should be used up to optimize the video delivery performance.
\end{proof}

Based on Theorem 2, problem P4 can be further transferred to the following feasible problem
\begin{equation} \label{p5}
\begin{aligned}
\text{\textbf{P5}}~~~~~& {\max\limits_{\mathbf{B}} }~~~~~V^*\\
\text{s.t.}~~~~~&\text{C1}:~\theta_n=-\frac{\ln\epsilon_n}{V^*T(d_n-d_n^{C}(1-\sum\limits_{f_m \in \mathcal{M}}\alpha_{n,m}\beta_m))}\\
&\text{C2}:~{\sum\limits_{U_n \in \mathcal{N}}B_n= B}\\
&\text{C5}:~V^*=-\frac{\ln\mathbb{E}[e^{-\theta_n R_n(i)T}]}{\theta_n T},~\forall U_n \in \mathcal{N}\\
\end{aligned}
\end{equation}

Problem P5 can be solved by using the monotonicity between $V_n$ and $B_n$.
Specifically, when $V^*$ is set with a low value, the network bandwidth cannot be used up.
In this case, $V^*$ can be raised.
On the other hand, if $V^*$ is set with a high value, the network bandwidth will be insufficient.
In this case, $V^*$ should be lowered.
With this regard, the optimal $V^*$ can be obtained by bisection exploration.
Additionally, when $V^*$ is fixed, the allocated bandwidth for each user can also be obtained by bisection exploration.
Let $\varphi_V$ and $\varphi_B$ denote the accuracy requirement of the video coding rate and bandwidth allocation.
Algorithm 2 summarizes a double bisection exploration scheme to find out the optimal video coding rate and the bandwidth allocation.
It is easily verified that the computation complexity of Algorithm 2 lies in $O(N\log_2(V^{\max}-V^{\min})\log_2B$).

\begin{algorithm}[t!]
\caption{Double Bisection Exploration Scheme}
\label{alg::conjugateGradient}
\begin{algorithmic}[1]
\State \textbf{Input} $\{\alpha_{i,j}\}$, $\beta$, $B$, $\{d_n\}$, $\{\epsilon_n\}$, $N$, $T$, $\varphi_V$, $\varphi_B$;
\State \textbf{Output} $V^*$ and $\mathbf{B}$.
\State Initialize $B_n=\frac{B}{N},~\forall U_n \in \mathcal{N}$;
\State Calculate $\{V_n\}$ according to (\ref{Vn}) and (\ref{theta});
\State Let $V^{\max}=\max\{V_n\}$ and $V^{\min}=\min\{V_n\}$;
\Repeat
    \State Set middle point $V^{\text{mid}}=\frac{V^{\max}+V^{\min}}{2}$;
    \For{$U_n \in \mathcal{N}$}
        \State Compute $\theta_n$ according to (\ref{theta});
        \State Let $B_n^{\max}=B$ and $B_n^{\min}=0$;
        \Repeat
            \State Set middle point $B_n^{\text{mid}}=\frac{B_n^{\max}+B_n^{\min}}{2}$;
            \State Calculate $\gamma=-\frac{\ln\mathbb{E}[e^{-\theta_n R_n(i)T}]}{\theta_n T}$;
            \If{$V^{\text{mid}}==\gamma$}
                \State Break;
            \Else
                \If{$V^{\text{mid}}>\gamma$}
                    \State $B_n^{\min}=B_n^{\text{mid}}$;
                \Else
                    \State $B_n^{\max}=B_n^{\text{mid}}$;
                \EndIf
            \EndIf
        \Until {$B_n^{\max}-B_n^{\min}\leq \varphi_B$ or $V^{\text{mid}}==\gamma$;}
        \State $B_n=B_n^{\min}$;
    \EndFor
    \If{$\sum_{U_n \in \mathcal{N}}B_n==B$}
        \State Break;
        \Else
            \If{$\sum_{U_n \in \mathcal{N}}B_n>B$}
                \State $V^{\text{max}}=V^{\text{mid}}$;
            \Else
                \State $V^{\text{min}}=V^{\text{mid}}$;
            \EndIf
    \EndIf
\Until {$V^{\max}-V^{\min}=\varphi_V$ or $\sum_{U_n \in \mathcal{N}}B_n==B$}
\end{algorithmic}
\end{algorithm}

\section{Experiments and Results}
In this section, extensive experiments are carried out to analyze the effectiveness of the proposed video caching and delivery schemes.
Particularly, a real data set, i.e., Movielens\cite{8523627}, is employed to validate the adaptability of our proposed video caching scheme.
The dataset contains the information including user attributes, the movie attributes as well as the movie ratings.
To represent the impact of user information on user interest, we apply Personas to mine useful user potential information from the statistical point of view.
Besides, labels are generated in terms of movie ratings.
Moreover, we randomly divide the training set and testing set for model training and interest prediction respectively.

In addition, according to the existing interaction information between users and movies, we generate the interactive information of movies that users have not watched before, and set the label as unknown.
We then randomly select $N$ users to perform experiment.
In addition, consider that the number of unwatched movies is much larger than that of the watch movies for a user, we randomly select $M=3E$ unwatched movies for each user to generate dataset $\mathbf{x}^{new}$.
If not otherwise highlighted, the various involved parameters for video delivery are as follows.
The video requests $\{\alpha_{n,m}\}$ can be obtained from $\mathbf{x}^{new}$.
For the wireless communication part, the transmission power of the AP for each user is set to 100dBm.
The length of a transmission block is set to 0.1s.
Besides, the network bandwidth is set to $50$MHz and the power spectral density of the background noise is set to $N_0=-130$dBm/Hz.
The channel power gain due to small-scale fading is assumed to follow exponential distributed with unit mean.
The path loss is assumed to be $l_n=\rho_{n}^{-2}$ with 30dB power attenuation at a reference distance of 1m, where $\rho_n$ denotes the distance between the AP and $U_n$.
Moreover, $\rho_n$ is randomly selected with range $[15,~20]$m based on Poisson point process (PPP).
For the wire communication part, the transmission delay between the cloud server and edge server is assumed to be 0.1s.

\subsection{Individual Interest Prediction Evaluation}
The performance of the proposed individual interest prediction model is first evaluated in terms of Area Under Curve (AUC) and model accuracy (ACC).
AUC is a general metric to evaluate the learning performance for binary classification problem.
The definition of AUC is as follows
\begin{equation}
\text{AUC} = P\{\widetilde{r}^{new}_{n_1,m_1} > {\widetilde{r}^{new}_{n_0,m_0}}\}
= \frac{{\sum {I({\widetilde{r}^{new}_{n_1,m_1}},{\widetilde{r}^{new}_{n_0,m_0}})} }}{{P*Q}},
\end{equation}
where
\begin{equation}
 {I({\widetilde{r}^{new}_{n_1,m_1}},{\widetilde{r}^{new}_{n_0,m_0}})} = \left\{ {\begin{array}{*{20}{c}}
{1,~\widetilde{r}^{new}_{n_1,m_1} > {\widetilde{r}^{new}_{n_0,m_0}}}\\
{0.5,~\widetilde{r}^{new}_{n_1,m_1} = {\widetilde{r}^{new}_{n_0,m_0}}}\\
{0,~\widetilde{r}^{new}_{n_1,m_1} < {\widetilde{r}^{new}_{n_0,m_0}}}
\end{array}} \right.
\end{equation}
where $P$ denotes the number of negative samples, $Q$ denotes the number of positive samples.

\begin{table}[h]
\begin{center}
\caption{Accuracy Evaluation Method}
\label{table}
\begin{tabular}{|l|l|l|l|}
\hline
\multicolumn{2}{|l|}{\multirow{2}{*}{}}       & \multicolumn{2}{c|}{Prediction} \\ \cline{3-4}
\multicolumn{2}{|l|}{}                        & 1          & 0          \\ \hline
\multicolumn{1}{|c|}{\multirow{2}{*}{Realistic}} & 1 & True Positive (TP)         & False Negative (FN)         \\ \cline{2-4}
\multicolumn{1}{|c|}{}                    & 0 & False Positive (FP)         & True Negative (TN)         \\ \hline
\end{tabular}
\end{center}
\end{table}

For the binary classification problem, we set the decision threshold as $\delta=0.5$, i.e., the prediction result is equal to 1 if the interest prediction score $\widetilde{r}^{new}_{n,m}$ is larger than 0.5, otherwise 0.
According Table I, the accuracy (ACC) of the proposed scheme on individual user interest prediction can be obtained as
\begin{equation}
ACC{\rm{ = }}\frac{{TP + TN}}{{TP + FN + FP + TN}},
\end{equation}
where $TP$, $FN$, $FP$, $TN$ denote the number of true positive samples, that of false negative samples, that of false positive samples and that of true negative samples respectively.

\begin{figure}[htbp]
\centerline{\includegraphics[scale=0.47]{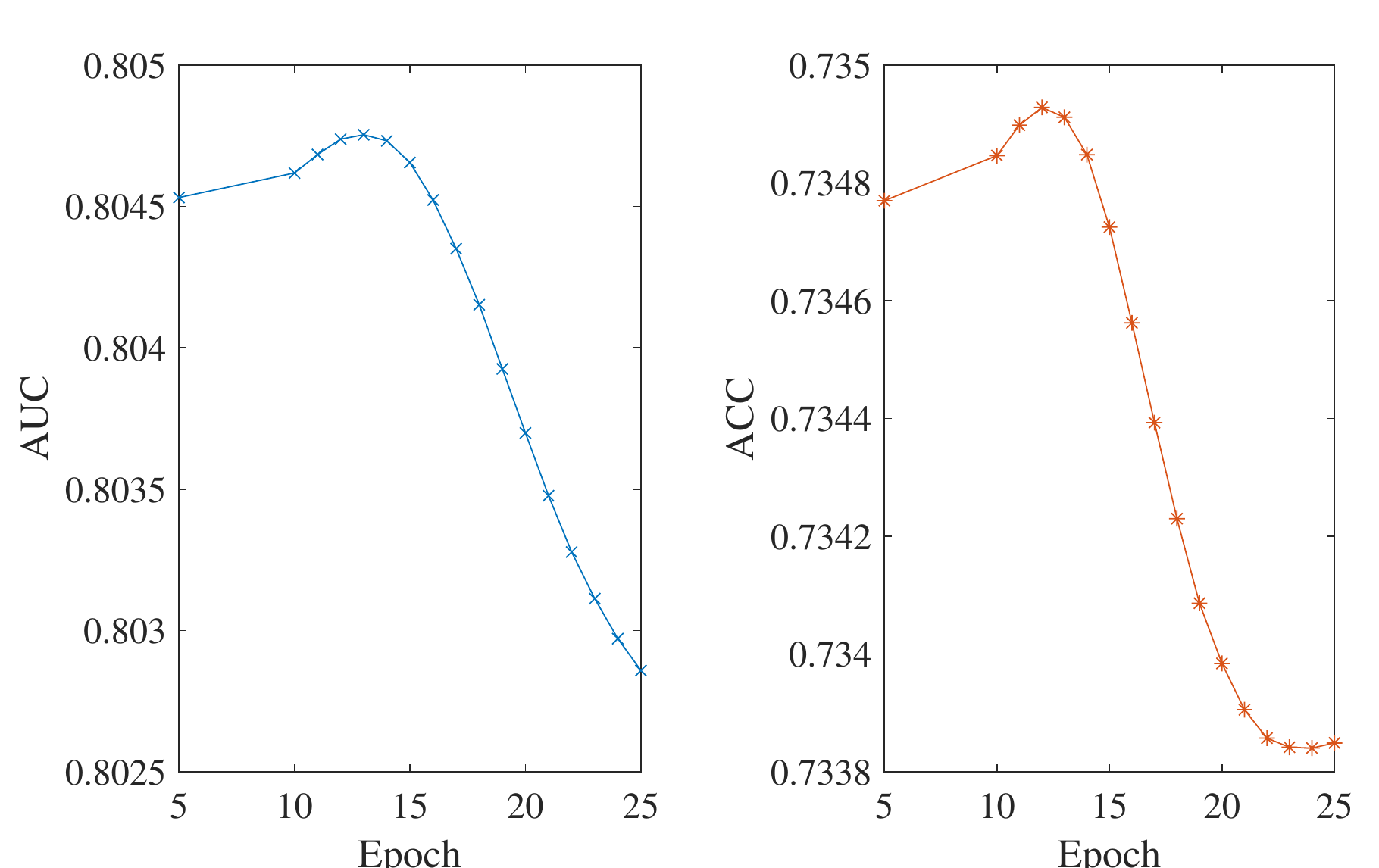}}
\caption{AUC and ACC of the proposed GIACF varying with epochs}
\label{fig3}
\end{figure}

Fig. \ref{fig3} depicts the performance of the proposed on user interest prediction in terms of AUC and ACC.
It is observed that there exists optimal epoch configuration (i.e., eopch=13) to obtain maximum AUC and ACC.
Additionally, it is verified that the proposed scheme can guarantee high AUC and ACC, which implies its effectiveness on individual user interest prediction.

\subsection{Video Caching Performance Evaluation}
In this subsection,  the performance of user hit rate (UHR) under the proposed scheme is first analyzed.
As mentioned earlier, the user hit rate is defined as the proportion of the users whose requests are hit at the edge server, i.e.,
\begin{equation}
\text{UHR} = \frac{\sum\limits_{U_n \in \mathcal{N}}\sum\limits_{f_m \in \mathcal{M}}\alpha_{n,m}\beta_m}{\sum\limits_{U_n \in \mathcal{N}}\sum\limits_{f_m \in \mathcal{M}}\alpha_{n,m}}.
\end{equation}
Additionally, two baseline schemes are considered here, which are the popularity-based scheme and the Top-K scheme \cite{liu2019deep}.
The popularity-based scheme caches contents according to the historical content popularity.
For the Top-K scheme, user interest is firstly predicted with our proposed FM and MLP merging scheme, and $E$ videos with highest individual interest prediction scores are then cached.

\begin{figure}[htbp]
\centerline{\includegraphics[scale=0.55]{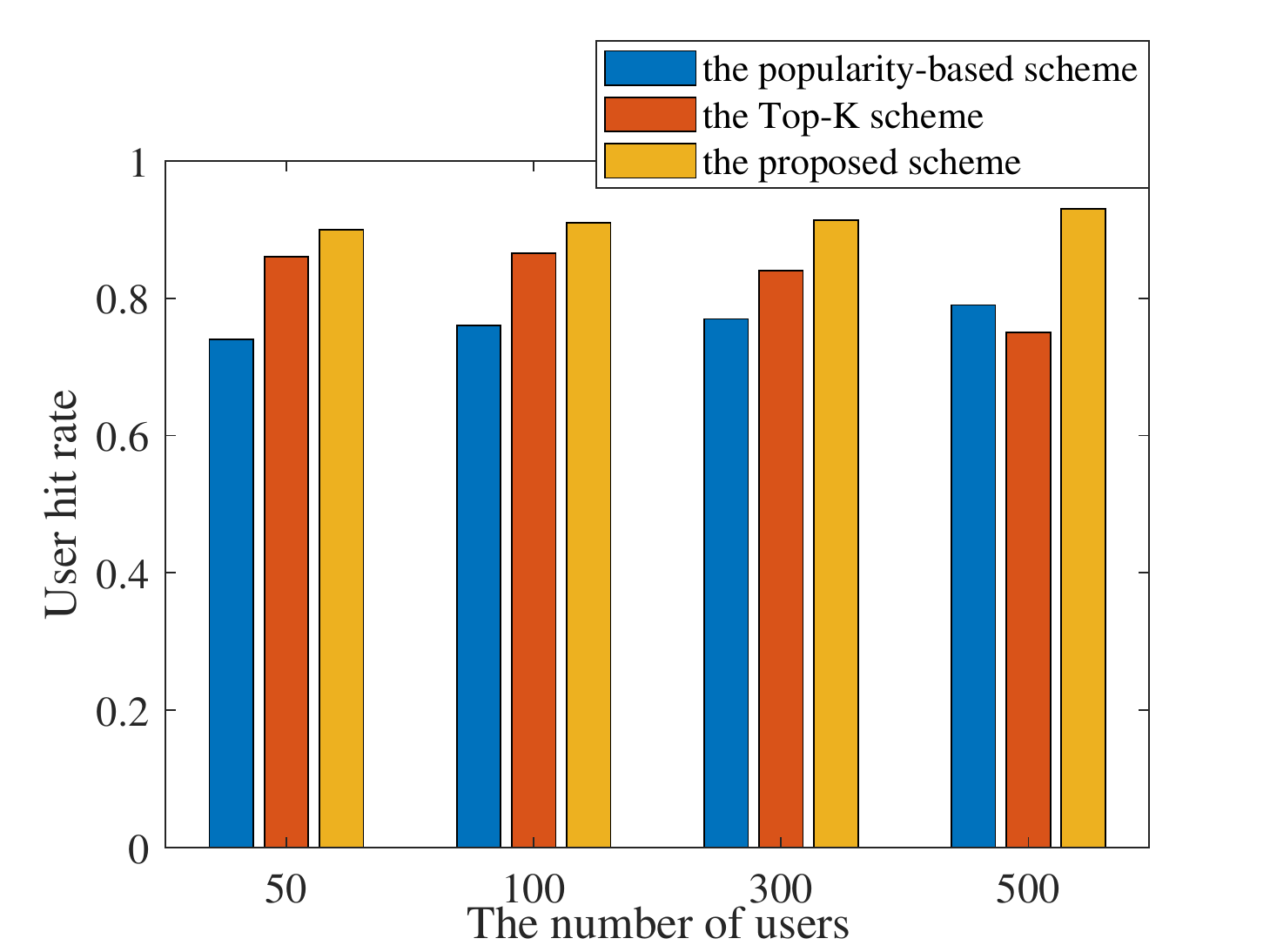}}
\caption{Comparison of different caching schemes on user hit rate}
\label{fig4}
\end{figure}

Fig. \ref{fig4} depicts the UHR performance under different schemes where the caching capacity is set as $E=100$.
It is observed that the UHR increases with the number of users under the popularity-based scheme.
It can be explained that the historical content popularity is equivalent to a group content popularity.
Differently, the Top-K scheme caches contents in terms of individual user interest.
If caching capacity is enough (e.g., $E=100$, $N\leq 100$), the UHR under the Top-K scheme is invariant.
As the number of users increases (e.g., $N\geq300$), however, the caching capacity is not enough to guarantee the individual user interest for each user.
At this time, the UHR of the Top-K scheme degrades, which implies that it is not applicable to apply individual user interest to make edge caching decision directly.
Additionally, the UHR under the proposed scheme is less sensitive to the number of users while comparing with the popularity-based scheme.
The reason is that the group interest of a community is statistically stable during a specific time.
As our proposed scheme can capture the group interest of a community accuracy, the UHR under the proposed scheme is stable.
Moreover, it is observed that the proposed scheme can guarantee a higher UHR than the other baseline schemes.
This validates the effectiveness of the proposed social aware similarity model on group interest characterization.

\begin{figure}[htbp]
\centerline{\includegraphics[scale=0.55]{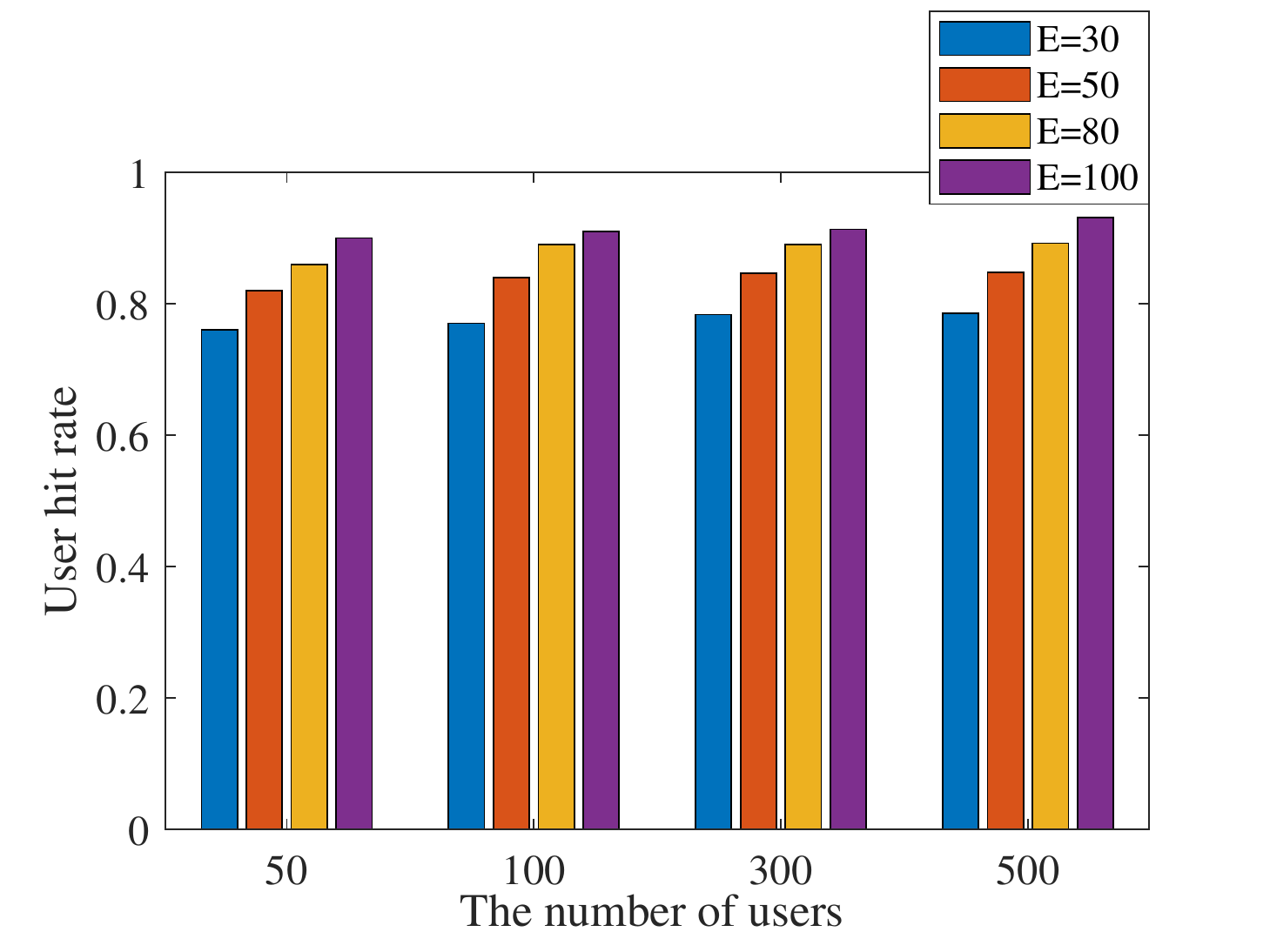}}
\caption{User  hit rate under different configurations of caching capacity and the number of users}
\label{fig5}
\end{figure}

Fig. \ref{fig5} depicts the UHR performance under different configurations of caching capacity $E$ and the number of users $N$.
It is found that the UHR increases as the caching capacity increases.
This is because caching more contents ensures a higher probability to hit the user requests.
Besides, increasing $N$ can bring gain to the UHR when $N$ is small (e.g., $N\leq100$), since the statistical information of group interest is not prominent in this case.
As $N$ increases (e.g., $N\geq100$), the statistical information of group interest turns to be stable, which implies the UHR is convergent.

In addition to user hit rate, content hit rate (CHR) is also an important metric to evaluate the content caching efficiency for the edge server.
Specifically, CHR is defined as the the proportion of cached contents that are requested by users, i.e.,
\begin{equation}
\text{CHR} = \frac{\sum\limits_{f_m \in \mathcal{M}}\min\{\sum\limits_{U_n \in \mathcal{N}}\alpha_{n,m},1\}\beta_m}{E}.
\end{equation}

\begin{figure}[htbp]
\centerline{\includegraphics[scale=0.55]{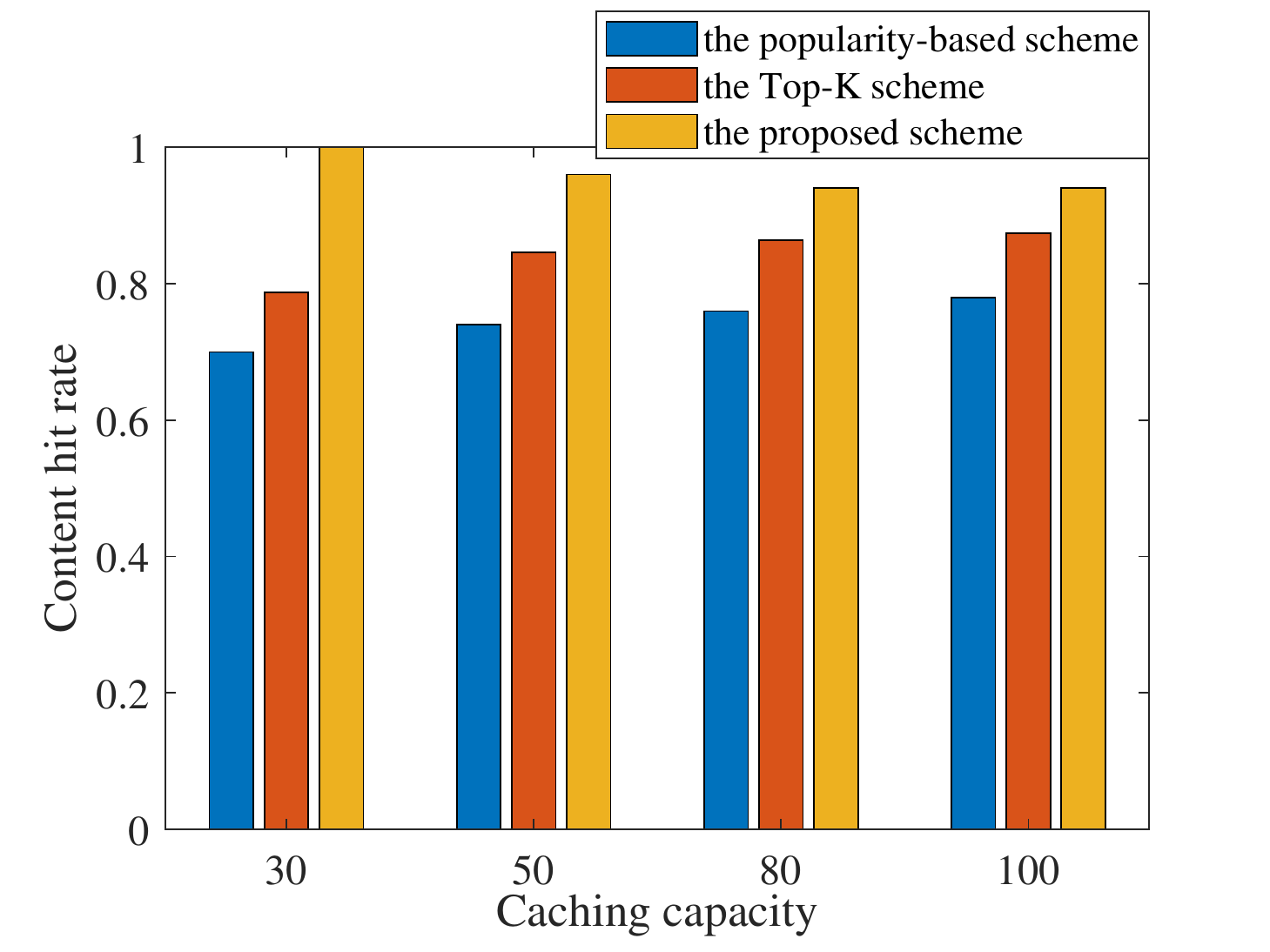}}
\caption{Comparison of different caching schemes on content hit rate}
\label{fig6}
\end{figure}

Fig. \ref{fig6} depicts the relationship between the content hit rate and the edge caching capacity, where the number of users is set to $N=100$.
It is verified that the proposed scheme outperforms the baseline schemes.
Besides, we observe that the CHR under the popularity-based scheme and that under the Top-K scheme both increase as the caching capacity increases.
This is because the popularity-based scheme caches videos according to video historical popularity.
Under the popularity-based scheme, some cached videos may be watched by the users before, which may degrade the CHR, since a user is usually willing to watch a video only once.
Hence, increasing the caching capacity will allow more videos to be cached, such that the CHR of popularity-based scheme can be improved.
The Top-K scheme caches videos based on the individual user interest only.
As one user may show strong interest on multiple videos, the Top-K scheme may cache large number of videos for some specific users within a caching decision duration.
A user, however, cannot watch too many videos within a caching decision duration.
Therefore, a larger caching capacity is able to allow the edge server to cache videos for more users whose interests are predicted to be comparably weaker.
Differently, the CHR under the proposed scheme deceases as the caching capacity increases.
Particularly, the proposed scheme can guarantee a nearly perfect CHR when edge caching capacity is small (i.e., $E=30$).
The reason is that the proposed scheme can acquire good representation of group interest.
As a result, the videos strongly associated to the group interest are predicted with high scores and cached accurately.
As the caching capacity increases, more videos that are not interested in by the users are cached, which results in CHR degradation and high economic cost.
However, from Fig. \ref{fig5}, larger caching capacity is able to guarantee higher UHR.
Hence, the tradeoff between UHR and CHR carefully should be taken into account.

\subsection{Video Delivery Performance Evaluation}
\begin{figure}[htbp]
\centerline{\includegraphics[scale=0.55]{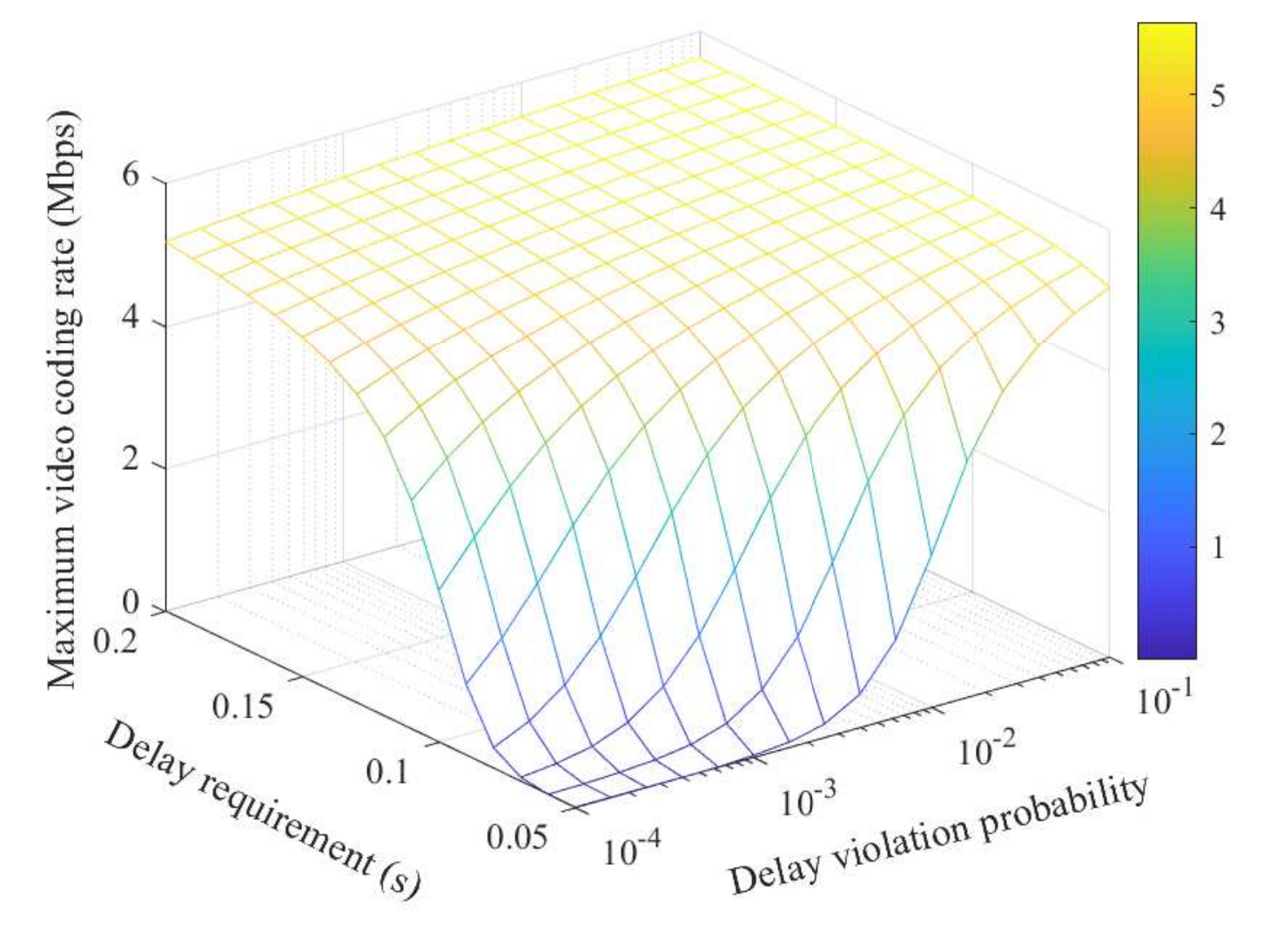}}
\caption{Individual video coding rate under different delay constraints}
\label{fig7}
\end{figure}

Fig. \ref{fig7} depicts the maximum individual video coding rate $V_n$ that can be sustained by a wireless link under different delay constraints.
The bandwidth of the link is set to 0.5MHz.
The distance between the AP and the user is set to 20m.
It is observed that $V_n$ increases with both the delay requirement $d_n$ and the tolerant delay violation probability $\epsilon_n$.
This is because larger $d_n$ or $\epsilon_n$ means looser delay constraint, which therefore allow higher video coding rate.
As the delay constraint is loosen enough, the sustained video coding rate turns to be convergent.
In this case, the video delivery is delay-tolerant and the video coding rate approaches to the mean channel capacity.

\begin{figure}[htbp]
\centerline{\includegraphics[scale=0.55]{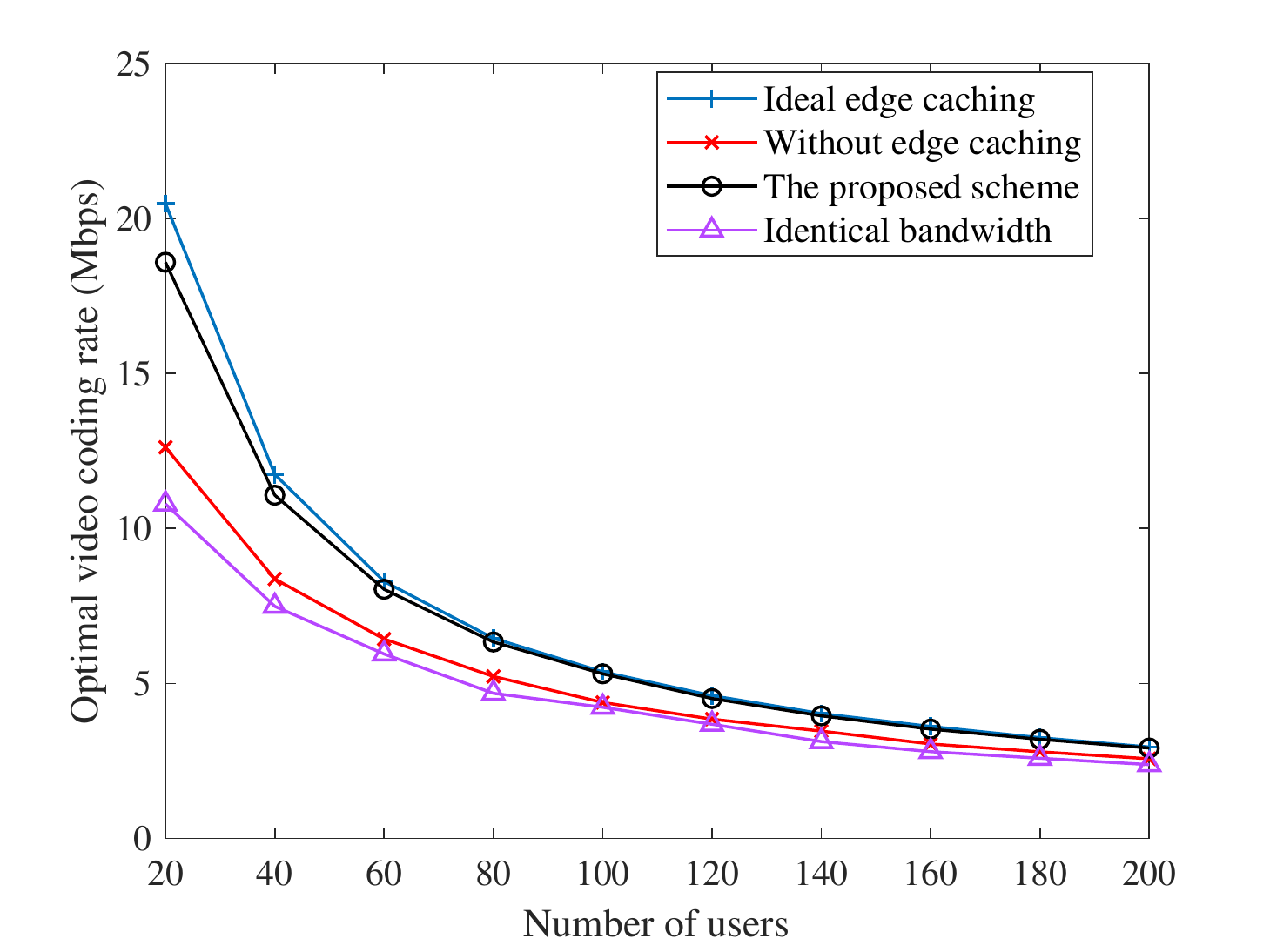}}
\caption{Optimal video coding rate under user fairness requirement}
\label{fig8}
\end{figure}

\begin{figure}[htbp]
\centerline{\includegraphics[scale=0.55]{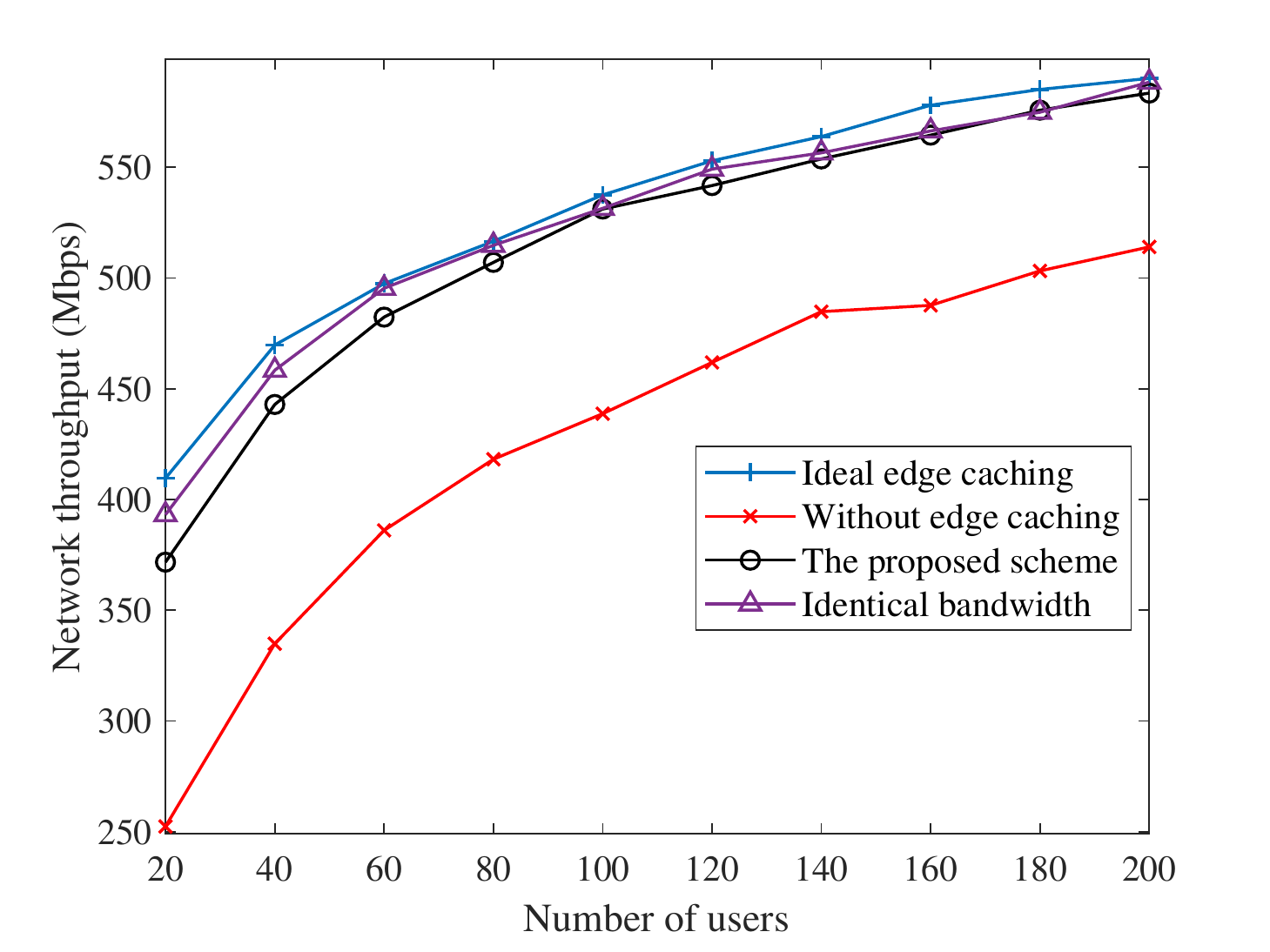}}
\caption{Network throughput under different number of users}
\label{fig9}
\end{figure}

Fig. \ref{fig8} and Fig. \ref{fig9} depict the performance of individual video coding rate and network throughput respectively.
The delay constraint is set to $d_n=0.2$s and $\epsilon_n=0.001$ for any $U_n \in \mathcal{N}$.
To verify the effectiveness of the proposed cloud-edge collaboration framework and that of the proposed video delivery scheme, three baseline schemes are introduced to perform comparison.
The ideal edge caching scheme assumes that all the requested videos are hit at the edge server, which avoids the transmission delay from the cloud server to the edge server.
The without edge caching scheme assumes that all the requested videos are delivery from the cloud server to users, which can be considered as a worst case.
In identical bandwidth scheme, we first apply the proposed caching scheme to cache videos first and then delivery videos to users with identical bandwidth.

In Fig. \ref{fig8}, we observe that the sustained video coding rate decreases as the number of users increases.
This is because the network bandwidth is fixed, more users means less allocatable resources for each user.
Besides, the performance of the proposed scheme is close to the ideal scheme.
The reason is that the proposed scheme caches video for users with high accuracy, most of the video requests can thus be responded by the edge server.
Compared with the scheme with edge caching, it is verified that deploying edge server can sustain higher video coding rate significantly, such that the quality of user experience can be improved.
Additionally, the proposed scheme is able to sustain higher minimum individual video coding rate than the scheme that allocates identical bandwidth to each user.
As a result, user fairness can be guaranteed by our scheme.

In Fig. \ref{fig9}, the network throughput is defined as the sum of video coding rate for each user.
It is observed that the network throughput increases with the number of users, which implies that all the schemes can achieve statistical multiplexing gain.
Besides, the network throughput performance under the identical bandwidth scheme is slightly higher than that under the proposed scheme.
It can explained that user fairness is guaranteed at the expense of performance degradation of some users that are with good channel quality.

\section{Conclusions}
In this paper, a video service enhancement strategy was studied under the proposed edge-cloud collaboration framework.
An optimization problem was first formulated to guarantee the user fairness in terms of video coding rate.
A hybrid human-artificial intelligence scheme was then proposed to help the cloud server to make video caching decision.
In addition, we proposed a double bisection exploration scheme to help the edge server to carry out video delivery management while statistical delay constraints are guarantee at the same time.
Experiment results validated the proposed video caching scheme and video delivery scheme were both able to guarantee high performance compared to other baseline schemes.
The analysis in this paper sheds new insight on video service enhancement through group interest mining and statistical delay guarantee.

\bibliographystyle{IEEEtran}
\bibliography{TMM}

\end{document}